\documentclass[11pt,a4paper]{article}

\usepackage{fullpage}
\usepackage{amsthm}
\usepackage{amsmath}
\usepackage{amssymb}
\usepackage{latexsym}
\usepackage{epsfig}
\usepackage{cite}
\usepackage{psfrag}
\usepackage{ifdraft}
\usepackage{color}
\usepackage{paralist}
\usepackage{graphicx}
\graphicspath{{figures/}} 

\usepackage[colorlinks,urlcolor=blue,citecolor=blue,linkcolor=blue]{hyperref}

\newcommand{\rdots}{\mathinner{%
  \mkern1mu\raise1pt\hbox{.}%
  \mkern2mu\raise4pt\hbox{.}%
  \mkern2mu\raise7pt\vbox{\kern7pt\hbox{.}}\mkern1mu}}

\newtheorem{theorem}{Theorem}[section]

\newtheorem{corol}[theorem]{Corollary}

\newtheorem{proposition}[theorem]{Proposition}

\newcommand{\histogram}{mountain}
\newcommand{\Histogram}{Mountain}
\newcommand{\pslg}{PSLG}
\newcommand{\Pslg}{PSLG}

\newcommand{\eqdef}{:=}
\newcommand{\etal}{et~al.}

\newcommand{\marrow}{\marginpar[\hfill$\longrightarrow$]{$\longleftarrow$}}
\newcommand{\niceremark}[4]
   {\textcolor{#4}{\textsc{#1 #2:} \marrow\textsf{#3}}}

\newcommand{\wolfgang}[2][says]{\niceremark{Wolfgang}{#1}{#2}{red}}

\begin{document}

\title{Memory-Constrained Algorithms for Simple Polygons\thanks{This work was initiated at the 
 Dagstuhl Workshop on
Memory-Constrained Algorithms and Applications,
 November 21--23, 2011. We are deeply grateful
to the organizers as well as the participants of the workshop for
helpful discussions during the meeting.
The results were presented
at the
28th European Workshop on Computational Geometry (EuroCG'12), in Assisi, Italy, March 2012~\cite{abbkmrs-mcasp-12}.}}

\author{
Tetsuo Asano\thanks{JAIST, Japan, \texttt{t-asano@jaist.ac.jp}} \and
Kevin Buchin\thanks{TU Eindhoven, Netherlands, \{\texttt{k.a.buchin, m.e.buchin}\}\texttt{@tue.nl } M. Buchin is supported by the Netherlands Organisation for Scientific Research
(NWO) under project no. 612.001.106. } \and
Maike Buchin\footnotemark[3] \and
Matias Korman\thanks{UPC, Barcelona. {\tt{matias.korman@upc.edu}}. With the support of the Secretary for Universities and
Research of the Ministry of Economy and Knowledge of the Government of Catalonia and the European Union.} \and
Wolfgang Mulzer\thanks{Freie Universit\"at Berlin,
 Germany, \{\texttt{mulzer, rote}\}\texttt{@inf.fu-berlin.de } } \and
G\"unter Rote\footnotemark[5] \and
Andr\'e Schulz\thanks{WWU M\"unster, Germany, \texttt{andre.schulz@uni-muenster.de} }
}

\maketitle

\begin{abstract}
A constant-work-space algorithm has read-only access
to an input array and may use only $O(1)$ additional
words of $O(\log n)$ bits, where $n$ is the input size.
We show how to triangulate a
plane straight-line graph with $n$ vertices in $O(n^2)$ time and constant work-space.
We also consider the problem of preprocessing a simple
polygon $P$ for shortest path queries, where $P$ is given by
the ordered sequence of its $n$ vertices.
For this, we relax the space constraint to allow $s$~words of work-space.
After quadratic preprocessing, the shortest path
between any two points inside $P$ can be found in $O(n^2/s)$ time.
\end{abstract}

\section{Introduction}
\label{sec:intro}

In algorithm development and computer technology, we observe
two opposing trends: on the one hand, there are vast amounts
of computational resources at our fingertips.
Alas, this often leads to bloated software that is written without regard to
resources and efficiency.
On the other hand, we see a proliferation of specialized tiny  devices
that have a limited supply of memory
or power. Software that is oblivious to space
resources is not suitable for such a setting.
Moreover, even if a small device features a fairly large memory, it may still
be preferable to limit the number of write operations.
For one, writing to flash memory is
slow and costly, and it also reduces the lifetime of the memory.
Furthermore, if the input is stored on a removable medium, write-access
may be limited for technical or security reasons.

With this situation in mind, it makes sense to focus on
algorithms that need only a limited amount
of work-space, while the input resides in read-only memory.
In this paper, we will develop such algorithms for geometric
problems in planar polygons or, more generally, plane straight-line graphs
(PSLGs).

In particular, we consider two fundamental problems from computational
geometry~\cite{deBergChvKrOv08}:
first, we are given a \pslg\ $G = (V,E)$ with $n$ vertices, and
we would like to find a \emph{triangulation} for $G$, i.e., a
\pslg\ with vertex set $V$ that contains all the edges in $E$ and
to which no edge can be added without violating planarity.
We show how to find such a triangulation in $O(n^2)$ time
with $O(1)$ words of work-space (Section~\ref{sec:polygon}).
Since our model does not allow storing the output, our algorithm outputs
the triangles of the triangulation one after another.

Then,  we apply this result in order to construct
a \emph{memory-adjustable} data structure for shortest path queries
in simple polygons (Section~\ref{sec:memadj}).
Given a simple polygon $P$ with $n$ vertices and a parameter
$s \in \{1, \ldots, n\}$, we  build a data structure for $P$ that
requires $O(s)$ words of storage and that lets us output the
edges of a shortest path between any two points inside $P$ in
$O(n^2/s)$ time using $O(s)$ work-space.
The preprocessing time is $O(n^2)$.

\paragraph{Model Assumptions.}\label{sec:def}

The input to our algorithms is either a simple polygon $P$ or a \pslg\
$G$
with $n$ vertices, stored in a read-only data structure.\footnote{A
\emph{plane straight-line graph} (PSLG) consists of a planar point set $V$
(vertices) and a set $E$ of non-crossing line segments
 with endpoints in $V$ (edges). By planarity, we have $|E| = O(|V|)$.}
In case of a \pslg, we assume that $G$ is given in a way
that allows us to enumerate all edges of $G$ in $O(n)$ time and to find the
incident vertices of a given edge in constant time (a standard adjacency
list representation will do).
In case of a polygon, we require that the vertices of $P$ are stored according
to their counterclockwise order along the boundary, so
that we can obtain the (clockwise and counterclockwise) neighbor
of any given vertex in constant time.
We also assume that is takes constant time to access the $x$-
and $y$-coordinates of any  vertex
and to perform basic geometric operations, such as determining
whether a point lies above or below a given line.

Storage is counted in terms of \emph{cells} or \emph{words}.
As usual, a word is assumed to be large enough to contain either
an input item (such as a point coordinate)
or a pointer into the input structure (of $\Theta(\log n)$ bits).
Thus, in order to convert our storage bounds into
bits, we have to multiply them by a factor of $\log n$.
In addition to the input, which can only be read, the algorithm has
$O(s)$ words of \emph{work-space} at its disposal for reading and writing.
Here, $s$ is a parameter of the model and can range between $1$ and $n$.
We will consider both the case where $s$ is a fixed constant
and the case where $s$ can be chosen by the user. Since there is no way
to store the result, we use an additional operation \texttt{output} in order
to generate output data. We require that every feature
of the desired structure is \texttt{output} exactly once.

For simplicity, we will make the usual general position assumption:
no three input vertices are on a line and no two input vertices have
the same $x$-coordinate.

\paragraph{Related Work.}
Given the many applications of memory-constrained algorithms, a
significant amount of research has been devoted to them,
even as early as in the 1980s~\cite{MunroPa80}.
One of the most studied problems in this setting is that of
selection in an unsorted array with elements from a totally ordered
universe~\cite{MunroRa96,MunroPa80,Frederickson87,RamanRa98,Chan10}.

In computational complexity theory, the constant-work-space model
is represented by the
complexity class LOGSPACE~\cite{AroraBa09}. There are several algorithmic
results for this class, most prominently Reingold's celebrated method for
finding a path between two vertices in an undirected
graph~\cite{Reingold08}. However, complexity theorists typically do not
try to optimize the running time of their constant-work-space algorithms,
whereas one of our objectives is to solve a given problem as quickly as
possible under the memory constraint.

There are other models that allow only read-access to the input, such as
the \emph{streaming} model~\cite{Muthukrishnan05} or the
\emph{multi-pass} model~\cite{ChanCh07}.
In these models, the input can be read only a bounded number of times
in sequential order, whereas we allow the input to be accessed in any order and
as often as necessary.
Other memory-constrained models are \emph{succinct}
data structures~\cite{JacoSucc89} and \emph{in-place}
algorithms~\cite{BronnimannIaKaMoMoTo04,BronnimannChCh04,ChanCh10,BronnimannCh06}.
The aim of succinct data structures is to use the minimum number of
bits to represent a given input. Although this kind of approach
significantly reduces the memory requirement,
in many cases $\Omega(n)$ bits are still necessary.
For in-place algorithms, we also assume that only $O(1)$ cells of work-space
are available. However, in this model we are  allowed to reorder and sometimes
rewrite the input data. This makes it possible to encode our data
structures through appropriate permutations of the input and often
to achieve the best possible running time. Several classic geometric problems,
such as convex hull computation or nearest-neighbor search, have
been considered in this
model~\cite{BronnimannIaKaMoMoTo04,BronnimannChCh04,ChanCh10,BronnimannCh06}.
Note that the improved running times in the in-place model come at
the expense of requiring more powerful operations from the computational
environment, making the results less widely applicable.
Moreover, in most of
these cases the input values cannot be recovered after the algorithm
is executed.

A classic algorithm from computational geometry that fits into
our model is the gift-wrapping method
(also known as Jarvis' march):
given $n$ points in the plane, we can report the $h$ points
on the convex hull in $O(n{h})$ time using $O(1)$ cells of
work-space~\cite{s-chc-04}.
More recently, Asano \etal~\cite{AsanoMuRoWa11} initiated
the systematic study of constant-work-space algorithms
in a geometric context.  They
describe algorithms to compute well-known geometric structures (such as the
Delaunay triangulation, the Voronoi diagram, and the Euclidean
minimum spanning tree) using $O(1)$ cells of work-space.
They also show how to obtain a triangulation of a planar point
set and how to find the edges of the shortest path between any two points
inside a simple polygon. These algorithms use constant work-space
and run in quadratic time.
We emphasize once again
that since the output may have linear size, it is not stored,
but reported piece by piece.
Recently, Barba \etal~\cite{BarKorLanSil11}
gave a constant-work-space algorithm
for computing the visibility
region of a point inside a simple polygon.

Although we know of no previous method to triangulate a given
simple polygon $P$ with sublinear work-space, it is known how
to find an \emph{ear} of $P$ in linear time and constant
space~\cite{ElGindyEvTo93} (an ear is a triangle inside $P$
defined by a single line segment between two vertices of $P$).
However, due to the limited work-space, there seems to be no easy way
to extend this method in order to obtain a complete triangulation of $P$.

\section{Preliminaries}
\label{sec:primitive}

We start by describing two simple operations and a basic spatial
decomposition technique that will be useful in designing our algorithms

\paragraph{Angular and Translational Sweep.} Let $s$ be a point in the
plane, and let $r$ be a ray with initial point $t\in r$ such that the line
segment $st$ does not intersect any edge of the input \pslg\ $G$.
In  an  \emph{angular sweep}, we move $t$ along $r$, and we would like
to determine the first vertex of $G$ that is hit by the line segment
$st$.
Note that the angular sweep is very similar to what is commonly called
a \emph{gift wrapping query}.
A \emph{translational} sweep is similar, but now we have a second ray
$r'$ with initial point $s$, and the segment $st$ is
swept vertically,  while maintaining the endpoints on $r$ and $r'$. Figure~\ref{fig:sweeps} depicts
an example of both sweep types.
In both cases, the first vertex of $G$ that is hit by $st$ can be found in
$O(n)$ time: we only need to check for each vertex whether it lies
in the swept area and to maintain the point of smallest angle or of smallest
horizontal distance, respectively.

\begin{figure}[htbp]
\begin{center}
\begin{tabular}{cp{8mm}c}
\includegraphics[scale=.80]{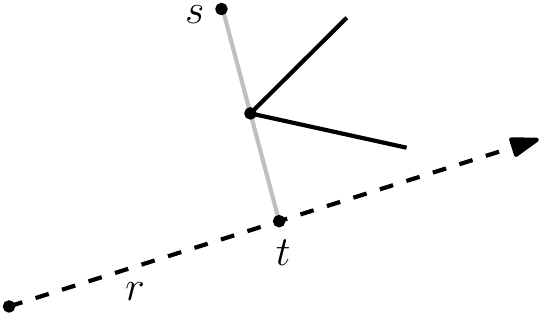} &&
\includegraphics[scale=.80]{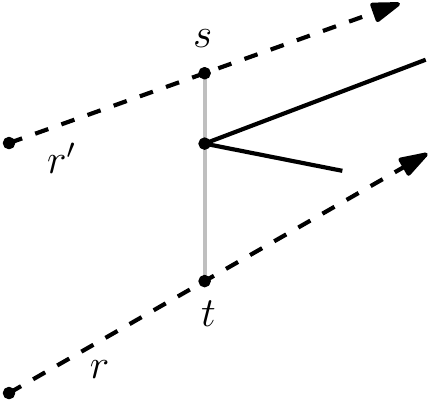} \\
(a) & &(b)
\end{tabular}
\caption{Illustration of an angular sweep (a) and a translational sweep (b).}
\label{fig:sweeps}
\end{center}
\end{figure}


\paragraph{Vertical Decomposition.}
The \emph{vertical
decomposition} (or \emph{trapezoidation}) of $G$ is
obtained by shooting vertical rays upward and downward from each vertex
of $G$ until they hit another edge of $G$ (or extend to
infinity)~\cite{deBergChvKrOv08}.
This gives a decomposition of the plane into trapezoidal cells.\footnote{Some
of these cells may be unbounded or degenerate into triangles.}
In general position, there are at most two vertices of $G$ on the
boundary of each cell.
(If $G$ represents a simple polygon~$P$, we may be interested
only in the trapezoids interior to $P$.)

In our model,
the trapezoids
incident to a given vertex $p$
can be determined easily in $O(n)$
time per trapezoid~\cite{AsanoMuRoWa11,AsanoMuWa11}: first, we enumerate
the edges of $G$ to find the first edge hit by the upward
and the downward vertical ray from $p$ in $O(n)$ time. Then,
we enumerate all edges incident to $p$ in circular order, including the
upward and the downward vertical ray. For each pair of consecutive edges, we perform an appropriate
translational sweep to find the trapezoid that is bounded by them. This takes $O(n)$ time per
edge.


\section{Triangulating a \Histogram{}}
\label{sec:hist}

In the next section, we will present an algorithm for
triangulating an arbitrary \pslg. First, however,
we need to explain how to handle inputs of a special kind.
The algorithm from this section will serve as an important building
block for the general case.

Let $H$ be a simple polygon with vertices
$a_1, a_2, \dots, a_k$, in circular order.
We call $H$ a  \emph{monotone mountain}\footnote{Also known as
\emph{unimonotone polygon}~\cite{FournierMo84}.}
(or \emph{\histogram} for short) if the $x$-coordinates of
$a_1, a_2, \dots$ increase monotonically.
The edge $a_1a_k$ is called the \emph{base} of $H$.
The \emph{shortest path} between two points $s$ and $t$ in $H$ is
the shortest polygonal chain with endpoints $s$ and $t$ that does not
cross the boundary of $H$.
We define the \emph{shortest path tree} SPT as the union of all
shortest paths from $a_1$ to the other vertices of $H$, see
Figure~\ref{fig:SPT}. SPT is a rooted tree with root $a_1$,
and it has the following properties:
\begin{proposition}
\label{prop:spt}
Any two adjacent edges of SPT form a left turn \textup(wrt.\ $a_1$\textup);
i.e., SPT bends only ``upwards''.
Let $f$ be an interior face of the \pslg\ formed by SPT
and $H$. Then (i) $f$ is bounded by the shortest paths
from $a_1$ to two consecutive vertices $a_i$
  and $a_{i+1}$;\footnote{This holds in any simple polygon.}
  (ii) $f$ is a
pseudotriangle\footnote{A pseudotriangle is a polygon with triangular convex hull.}, bounded from below by
 an SPT edge $ua_{i+1}$, from the right by an edge $a_ia_{i+1}$ of $H$,
 and from above by a concave chain
 of SPT edges \textup(as seen from inside\textup);
 and (iii)
 $f$ can be triangulated uniquely by connecting
  the rightmost vertex $a_{i+1}$ with the
  reflex vertices on the upper
  boundary.\qed
\end{proposition}

\begin{figure}
\begin{center}
\includegraphics{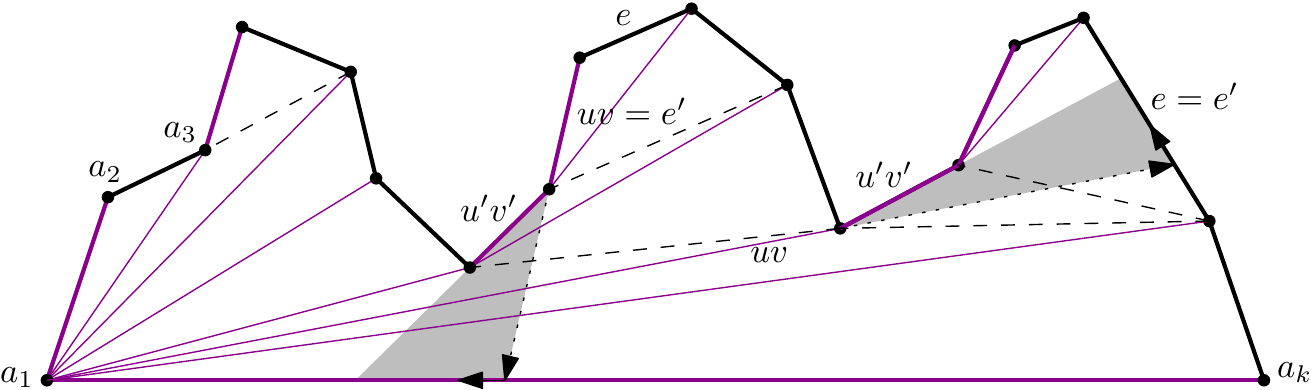}
\end{center}
\caption{The shortest path tree SPT from $a_1$ to all other vertices
(in purple). Additional triangulation edges generated in our algorithm
are dashed. The figure illustrates a forward step (right shaded area)
and a backward step (middle area).}
\label{fig:SPT}
\end{figure}

The idea now is to generate the triangulation during a depth-first traversal of
the edges of SPT, starting
from the base edge $a_1a_k$ and visiting the children
of each vertex in \emph{counterclockwise} order.
This traversal can be interpreted geometrically as an Euler tour
of the plane graph SPT.
Since there is no space for a stack, this tour must be performed
in a ``stateless'' manner, using angular sweeps to determine the
next edge to be traversed.
We call a vertex \emph{finished} if it has been visited by
the tour and will not be visited again. Otherwise, the vertex is
\emph{unfinished}. In an Eulerian traversal of SPT, the vertices
of $H$ become finished in order from right to left.

Our algorithm maintains two edges: (i) the current edge of the tour $uv$,
with $v$ lying to the right of $u$; and
(ii) the edge $e=a_ia_{i+1}$ of $H$ such that
$\{a_{i+1},a_{i+2}\ldots,a_k\}$ are the finished vertices of the tour.
Observe that we can use
 $e$ to distinguish between vertices that
are finished and those that are not.
In each step we distinguish three different cases,
and we accordingly perform a step as follows.

\textbf{Case 1:}
 \emph{$v$ is not incident to $e$}.
We perform a \emph{forward} step
into the subtree rooted at $v$.

\textbf{Case 2:} \emph{$v=a_{i}$,
but $uv$ is a chord of $H$}.
We perform a \emph{sideways} step to the
next edge out of $u$
that follows $uv$ in counterclockwise order.

\textbf{Case 3:} \emph{$v=a_{i}$, 
and $uv$ is the 
 edge $a_{i-1}a_i$ of $H$}. We do a \emph{backward} step and
return to the parent of $u$.

\begin{figure}
\begin{center}
\begin{tabular}{ccc}
\includegraphics[width=.29\textwidth]{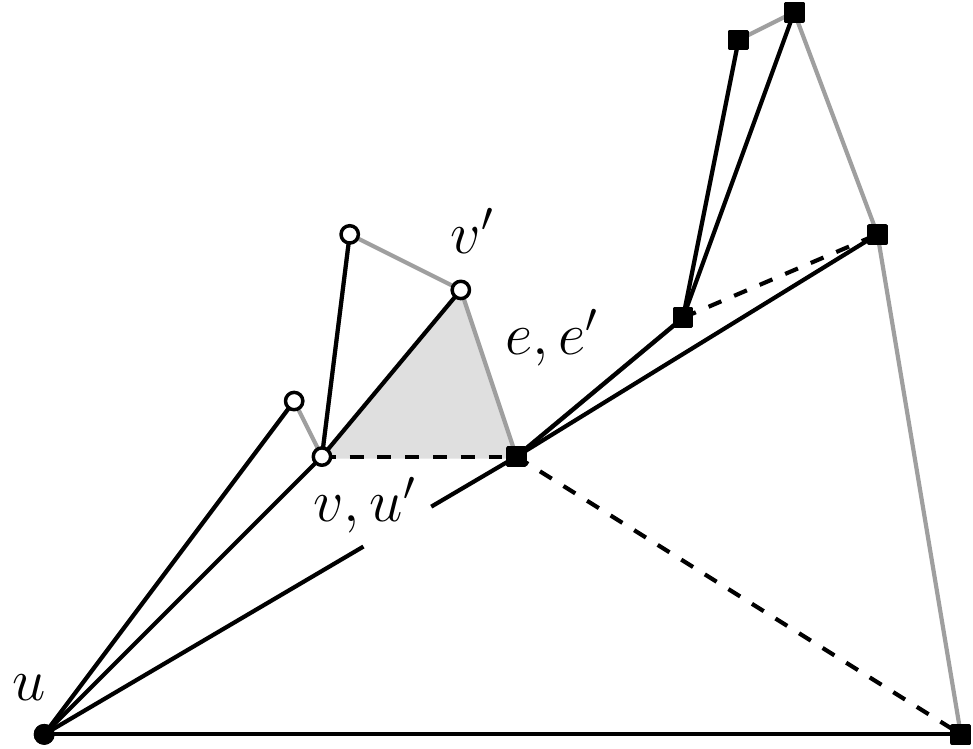} &
\includegraphics[width=.29\textwidth]{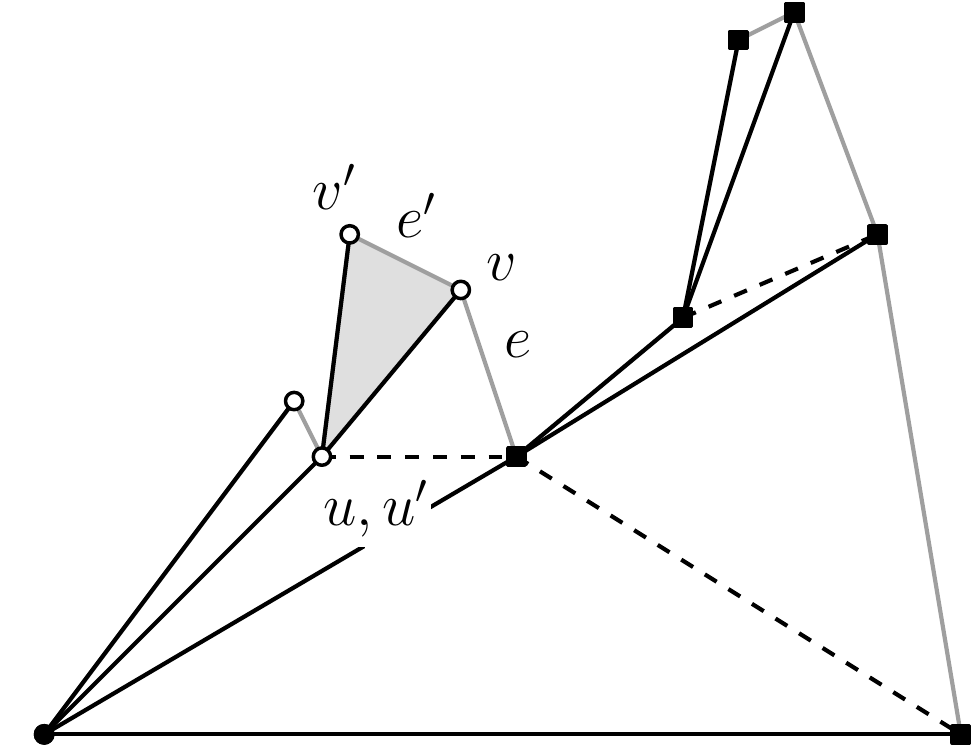} &
\includegraphics[width=.29\textwidth]{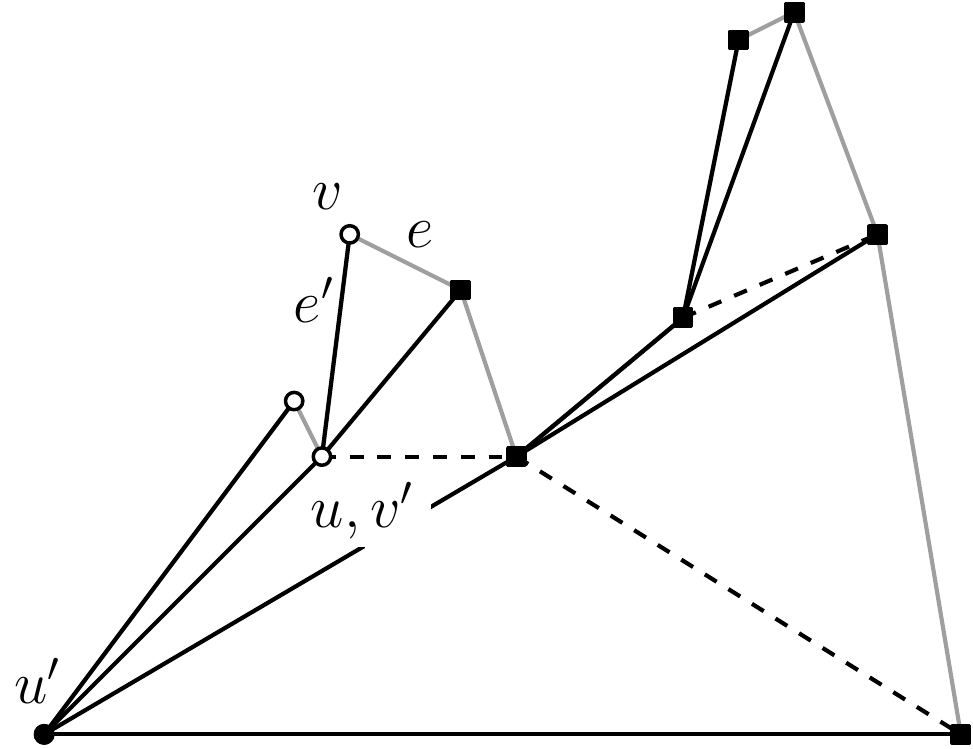}
\\
forward step & sideways step & backward step
\end{tabular}
\end{center}
\caption{The three cases for the Eulerian traversal of SPT. The triangles that are reported while performing a step are shaded. Edges in the triangulation that are not contained in SPT are drawn dashed.}
\label{fig:unimonotone-traversal}
\end{figure}

\noindent
We start the algorithm with
 a sideways step from
$uv  \eqdef e \eqdef a_1a_k$ (as an exception to the above rules).
The algorithm continues until all vertices are finished and it
tries to make a backward step from $e=a_1a_2$.
The details of the three steps are straightforward.
In each step, we determine the
values $u'$, $v'$ and $e'$ for the next step, and we output some triangles
of the triangulation (see Fig.~\ref{fig:unimonotone-traversal}; see also
Fig.~\ref{fig:SPT} for the forward and the backward step).

\textbf{Forward step.}
Let $x$ be the intersection of the line $\overline{uv}$ with the
edge $e$. We perform a counterclockwise angular sweep of the segment
$vx$ around the vertex $v$, letting $x$ move along
$e=a_ia_{i+1}$ (By Proposition~\ref{prop:spt}, SPT makes only
upward bends, so the line $\overline{uv}$ intersects $e$ and
the segment $vx$ lies inside $H$).
Let $z$ be the first vertex hit by the sweep
(note that $z$ might be $a_i$).
Since $z$ is the first child of $v$,
we update $u'v' \eqdef vz$, leave $e' \eqdef e$ and
output the triangle $vv'a_{i+1}$ (by Proposition~\ref{prop:spt}(iii)).

\textbf{Sideways step.}
Since $a_i$ is now finished, we
proceed to the previous edge $e' := a_{i-1}a_i$.
We make a counterclockwise angular sweep of the segment $uv$,
around $u$, letting $v$ mode along $e'$.
Let $z$ be the first vertex that is hit (again, $z$ might be $a_{i-1}$.
We set $u'v' \eqdef uz$ and output the
triangle  $uvv'$ (by Proposition~\ref{prop:spt}(iii)).

\textbf{Backward step.}
Since $a_i$ is now finished, we
proceed to the previous edge $e' := a_{i-1}a_i = uv$.
Let $x$ be the intersection of the line $\overline{uv}$
with the base edge $a_1a_k$. As before, Proposition~\ref{prop:spt}
ensures that $x$ exists.
We do a clockwise angular sweep of the segment $ux$ around
$u$, keeping $x$ on the base edge.
Call the first vertex that is hit $z$.
We set $u'v' \eqdef zu$.

Each edge of SPT is visited at most twice, so there are $O(k)$ steps.
Each step involves one angular sweep and some
additional processing that takes constant time. Thus, we get

\begin{theorem}
\label{thm:histo_triang}
Let $H$ be a \histogram{} with $k$ vertices. There is an algorithm
that does one counterclockwise scan of $H$
and outputs all triangles in a triangulation of $H$.
The algorithm needs $O(k)$ angular sweeps and
additional processing time $O(k)$ as well as constant work-space.
\end{theorem}
In particular, we have shown the following.
\begin{corol}
Let $H$ be an explicitly given \histogram{}\ with $k$ vertices.
Then we can output all the triangles in a triangulation of $H$
in $O(k^2)$ time with constant work-space.
\end{corol}

\paragraph{Remark.}
Our algorithm produces the same triangulation as the classic
algorithm for triangulating monotone polygons~\cite{GareyJoPrTa78}.
This algorithm processes the
vertices from left to right, and it
maintains a stack that represents
the lower convex hull of the vertices $\{a_1,\ldots,a_i\}$ encountered
so far.
This lower convex hull is also the shortest path from $a_1$ to $a_i$.
 When processing the next vertex $a_{i+1}$,
 the vertices that disappear from the hull are popped from the stack
 and appropriate triangles between the popped vertices and
 the new vertex are generated.

The classic algorithm also performs an implicit depth-first traversal of SPT,
but in contrast to our approach,
the children
of a vertex are visited in \emph{clockwise} order.
Even though we can modify this algorithm for constant work-space,
it does not perform a single scan of the vertices of $H$, a property that
will be crucial in the next section.

\section{Triangulating a \Pslg}
\label{sec:polygon}

We now describe how to triangulate
a \pslg\ $G$ with $n$ vertices in $O(n^2)$
time and with constant work-space.
Our algorithm only needs $O(n)$ scans over the edges of
$G$, and we make no assumptions about the scanning order.

In the usual model of computation, it is well known that $G$ can
be triangulated using the \emph{vertical decomposition}
(cf. Chazelle and Incerpi~\cite{ChazelleIn84}).
Since our algorithm follows the same strategy, we briefly review
how this works: first, we compute
the edges of the convex hull of $G$ and add them to $G$. Then
we find the vertical decomposition of the resulting graph,
and insert an edge between any two non-adjacent vertices of
$G$ that are contained in the same trapezoid
(see Fig.~\ref{fig:histo_decomp}).
Consider the resulting graph $G'$ that contains the convex hull edges
and the newly inserted edges.
All  interior faces of $G'$ are \histogram s:
since every vertex (except for the left- and the rightmost ones)
has at least one edge incident to either side,
the faces must be monotone polygons.
Suppose there is a face $f$ with vertices on both the upper and the lower
boundary. Then there would be
a trapezoid with a vertex at the upper and at the lower
boundary, and we would have inserted an edge between those vertices,
so $f$ cannot exist. Now, we can easily triangulate each face of the
resulting graph with an algorithm for triangulating \histogram s.

In our setting, we cannot explicitly compute the decomposition of $G$.
Instead, we enumerate all edges of $G$ and of the convex hull of
$G$. Note that convex hull edges can be found
in linear time per edge through Jarvis march~\cite{s-chc-04}.
For each such edge $e$, we check whether $e$ is the  base of a
\histogram{}.
This is the case if and only if
$e$ is incident to more than one trapezoid above it or more than
one trapezoid below  it.
As described in Section~\ref{sec:primitive}, this can
be checked in linear time.
(An inserted chord is never the base of a \histogram.)
Once the base of a \histogram{} $H$ is at hand, we would like to
use Theorem~\ref{thm:histo_triang} to triangulate it.
For this, we need  to enumerate the
vertices of $H$ in counterclockwise order and to perform
angular sweeps. The former can be done by enumerating the trapezoids
that are incident on one side of $e$, and takes $O(n)$ time
per trapezoid (see Section~\ref{sec:primitive}).
The latter can be done in $O(n)$ time by enumerating the vertices of $G$, because
the vertices outside $H$ cannot affect the angular scan.
Thus, by Theorem~\ref{thm:histo_triang}, it takes $O(nk)$ time to triangulate
$H$, where $k$ is the number of vertices in $H$.
Since the total size of
all \histogram s is $O(n)$, we get the following result.

\begin{theorem}
Given a \pslg\ $G$ with $n$ vertices,
we can output all triangles in a triangulation of $G$ in $O(n^2)$
time with constant work-space.
\end{theorem}

\begin{figure}
\begin{center}
\includegraphics[scale=0.55]{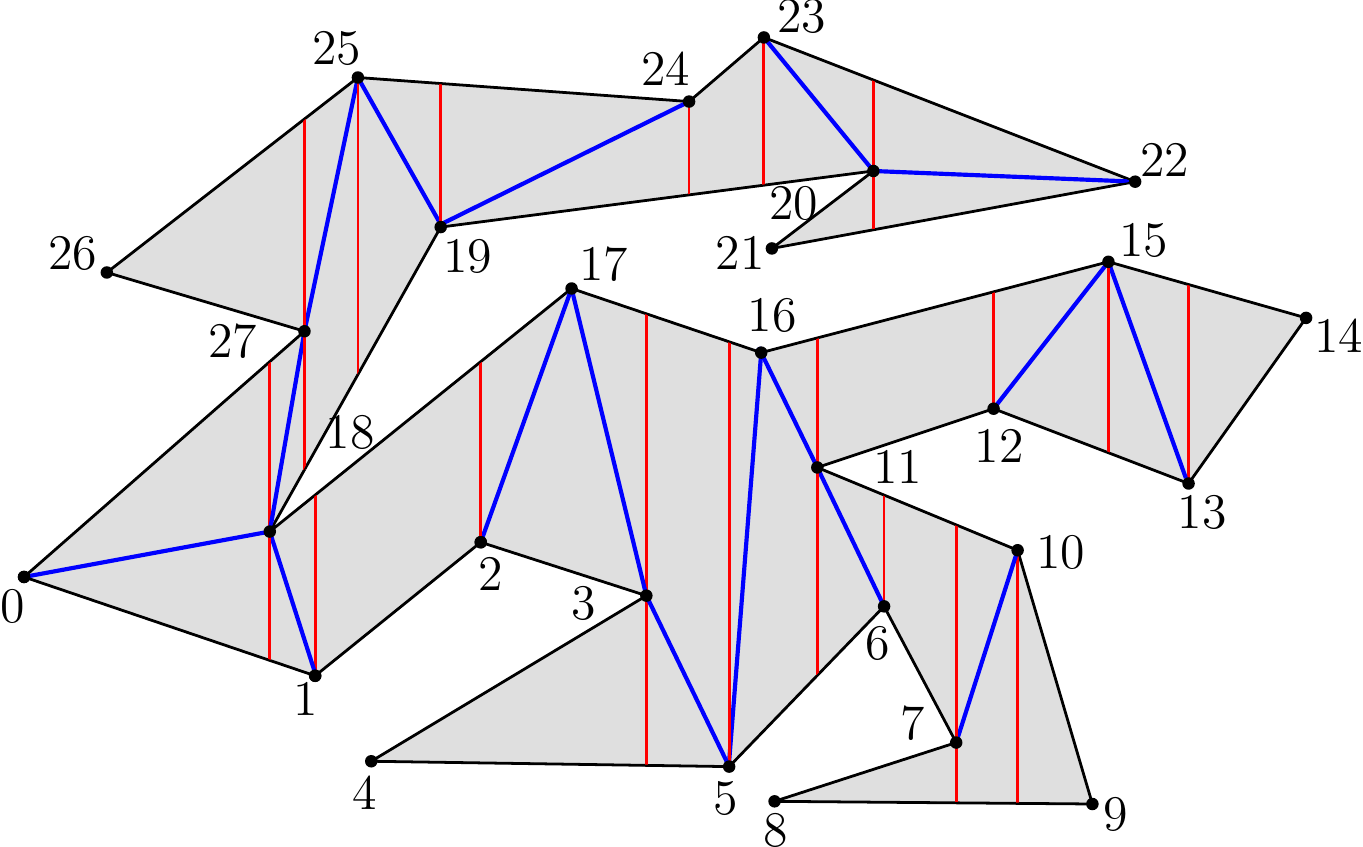}
\end{center}
\caption{Decomposition of a polygon into \histogram s.}
\label{fig:histo_decomp}
\end{figure}

\section{A Memory-Adjustable Data Structure for Shortest Paths}\label{sec:memadj}

Let $S$ be some set of $n$ objects.
In general, the purpose of a \emph{data structure} $D$ for
$S$ is to support certain \emph{queries} on
$S$ efficiently.
Ideally, $D$ has linear size,
and the query algorithm searches within $D$ with only $O(1)$ cells of
additional work-space.
In the classic setting, the whole set $S$ is contained in the
data structure,
so the storage must be at least as large as the input.

We take a different approach: recall that our input
cannot be modified. Thus, our strategy is to
preprocess the data and to store some additional information
in a data structure of size $s$ (for some parameter $s \leq n$).
The objective is to design an algorithm that uses this additional
information in a way that supports efficient query processing.
Ideally, we would like to have a trade-off between the amount of additional
storage and the running time of the algorithm. This has been done
successfully for many other classic problems such as selection and
sorting~\cite{MunroRa96,RamanRa98}.

Naturally, the most important quality measure for any such algorithm is the
query time. However, the preprocessing
time for constructing the data structure should also be taken into account. In
this section, we describe a data structure for computing the shortest path
between any two points inside a simple polygon $P$ with $n$ vertices.
It is known that this takes $O(n^2)$ time with constant
work-space~\cite{AsanoMuRoWa11}  and $O(n)$ time  with
linear work-space~\cite{GuiHers89}.

We describe how to construct a data structure that requires $O(s)$
words of storage and that can find the shortest paths between any two
points in $P$ in time $O(n^2/s)$. The preprocessing time is $O(n^2)$, and
the query algorithm is allowed to use $O(s)$ additional words of work-space.
Note that for $s=1$ the query time is $O(n^2)$, while for
$s = n$ the query time is $O(n)$. Thus, we achieve a smooth trade-off between
the known results.

\begin{figure}[htbp]
\centerline{\includegraphics[scale=0.72]{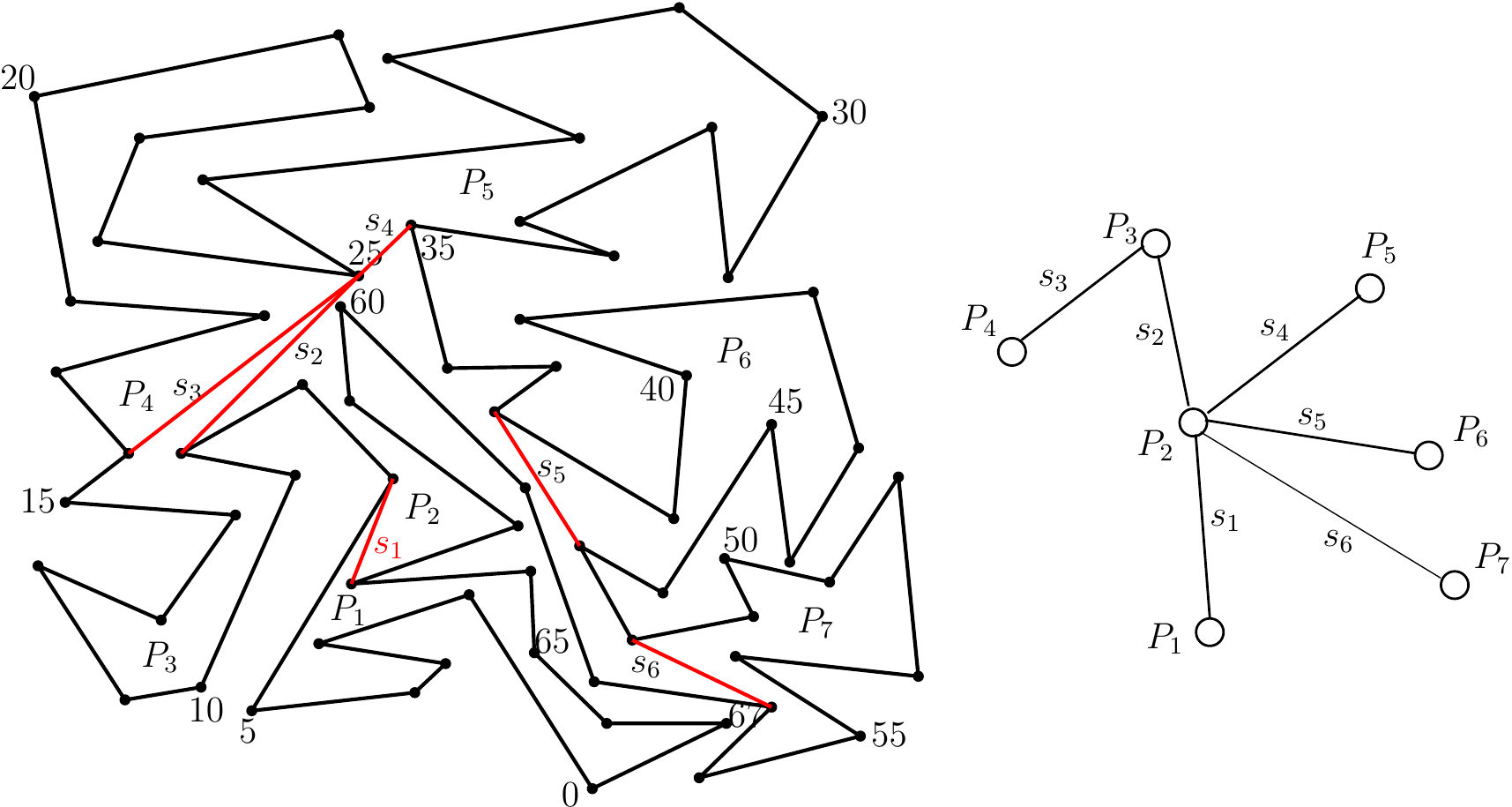}}
\caption{A decomposition of a simple polygon into
$7$ disjoint subpolygons using $6$ chords and its
associated tree decomposition.}
\label{fig:slab7}
\end{figure}

\subsection{Preprocessing}
During the preprocessing phase, we use our allotted space to
compute $O(s)$ chords that partition $P$ into $O(s)$ subpolygons with
$\Theta(n/s)$ vertices each, see Figure~\ref{fig:slab7}.
In the following, these chords will be called
\emph{cut edges}. Our data structure stores the subpolygons
together with a tree that represents the adjacencies between them.
Each subpolygon is represented by the corresponding cut edges and segments of
the boundary of $P$, in counterclockwise order. The boundary segments
can be described with $O(1)$ storage cells each, by using
the index of the first and the last edge. Thus, the total
space for the data structure is $O(s)$.


We use the following well-known observation~\cite{Chazelle82},
which is based on the fact that
any triangulation of a simple polygon dualizes to a tree with
$n-2$ nodes and
maximum degree $3$.
 \begin{proposition}
\label{prop:cut}
   Let $P$ be a simple polygon with $n \geq 4$ vertices.
   Any triangulation of $P$ has a chord that splits $P$ into
   two subpolygons with at most
   $\lfloor (2/3)n\rfloor + 1 \leq (5/6) n$ vertices each.
\qed
 \end{proposition}

A chord as in Proposition~\ref{prop:cut} is called a \emph{balanced cut}.

\begin{theorem}\label{thm:preprocess}
Let $P$ be a simple polygon with $n$ vertices. For any $s\leq n$,
there exists a set of $O(s)$ chords that are pairwise non-intersecting and that
partition $P$ into $O(s)$ subpolygons with $\Theta(n/s)$ vertices each.
The chords can be found in $O(n^2)$ time using constant work-space.
\end{theorem}
\begin{proof}
  We set $t := \lceil n/s \rceil$.
If $t \le 3$, we simply output a triangulation of $P$.
Otherwise, we use Proposition~\ref{prop:cut} iteratively to
split $P$ into smaller and smaller subpolygons.
In each round, we scan over all pieces computed so far,
and for every piece with more than $t$ vertices we find a balanced
cut.
In the end, each piece has between $t/6$ and $t$ vertices, so
we obtain $O(s)$ subpolygons of size $\Theta(n/s)$.

Let $Q$ be a subpolygon with $k$ vertices. We can find a balanced cut for $Q$
in time $O(k^2)$: triangulate $Q$ using Theorem \ref{thm:histo_triang}.
Whenever the algorithm outputs a new chord, we compute the size
of the two pieces it cuts, and we remember the one with
the most balanced cut. Note that the size of each piece can
be computed in $O(k)$ time: the  subpolygon $Q$ is represented as
a sequence of at most $O(k)$ cut edges or segments of the boundary of $P$,
and for each such segment we can determine the relevant size in
constant time.

Consider the $r$-th round, and  suppose we have $p\le n$ pieces of size
$n_1,\dots,n_p$, where $n_1+\dots+n_p=n+2p\le 3n$.
By Proposition~\ref{prop:cut},
we have $n_i \leq t_r := (5/6)^r n$ for all $i$.
Thus,  the running time for round $r$ is proportional to
\[
\sum _{i=1}^p n_i^2
\leq
t_r \sum _{i=1}^p n_i
\leq
t_r \cdot 3n =  3\Bigl(\frac{5}{6}\Bigr)^rn^2.
\]
Summing over all rounds, we get a total running time of $O(n^2)$, as claimed.
\end{proof}

\paragraph{Remark.}
Guibas and Hershberger~\cite{GuiHers89} showed that
if linear work-space is available, the preprocessing time can be reduced
to $O(n)$. For completeness, we briefly sketch their method.
First, we triangulate $P$ with Chazelle's
algorithm~\cite{Chazelle91}.
Then, we find the cut edges by greedily pruning the tree $T$ that
corresponds to the dual graph of the triangulation.
Set $t := \lceil n/s \rceil$.
Every vertex of $T$ has a weight, initialized to $1$.
In each round, we scan over the leaves of $T$. We remove those
leaves whose weight is between $t/3$ and $t+1$
and declare the corresponding edges to be cut edges. Then, we
delete the remaining leaves and add their weight to their parents.

In the end, the remaining part may have less than $t/3$ vertices. If so,
we remove one cut edge to merge this part with an adjacent one.
Since each vertex is visited once, the whole procedure takes linear time.

\subsection{Query Algorithm}
We now describe the query algorithm. Given two points $p ,q \in P$,
we would like to find the shortest path $\pi_{pq}$ between them,
while taking advantage of the precomputed polygon decomposition. The main idea
is to compute a shortest path for each subpolygon with the constant-work space
algorithm of Asano~\etal~\cite{AsanoMuRoWa11},
and to concatenate the resulting paths. However, additional steps are
necessary to deal with edges of $\pi_{pq}$ that cross several subpolygons,
so the algorithm gets a bit more involved.

First, let us quickly review the algorithm of
Asano~\etal~\cite{AsanoMuRoWa11} (referred to as AMRW from now on).
The AMRW-algorithm stores a triple $(v, r_1, r_2)$ of points.
The point $v$ is a vertex of $\pi_{pq}$, while $r_1$ and $r_2$ lie on
the boundary of $P$ (not necessarily vertices). The triple maintains the invariant
that all vertices of $\pi_{pq}$ up to $v$
have been reported. The line segments $vr_1$ and $vr_2$ cut off a
subpolygon $P' \subseteq P$ that contains the target $q$.
In each step, the algorithm shoots a ray into $P'$ that originates at $v$ and
that lies inside the \emph{visibility cone} determined by $vr_1$ and $vr_2$.
The direction of the ray is chosen according to a case distinction
whose details we omit. Let $r'$ be the point where this ray hits
the boundary of $P'$. The line segment $vr'$ divides
$P'$ into two parts, and by determining which part
contains the target $q$, we can find a new
triple $(v', r_1', r_2')$. This triple either yields a new vertex of $\pi_{pq}$,
or it makes $P'$ smaller.
AMRW show that in either case the ray shooting operation can be
charged to a vertex of $P$ in a way where every vertex is charged at most twice.
Thus, since each step takes linear time, the total
running time is bounded by $O(n^2)$. Please refer to
the original article for further details~\cite{AsanoMuRoWa11}.

Our algorithm uses a similar strategy: it also maintains a triple $(v,r_1,r_2)$
of points that fulfills the same invariant, and in each step it shoots a ray
to determine the triple for the next step. As long as the triple and the
ray are contained in a simple subpolygon of the decomposition, we can just
use the previous method without change while achieving the desired speedup.
We call this the \emph{standard situation}. However, if the points of the
triple are not contained in the same polygon, or if the ray crosses a cut edge
we need to take additional measures in order to quickly update the triple.
In this case, the algorithm switches to a different mode, the
\emph{long-jump situation}. We now describe the details.

\paragraph{Initialization.}
We start by locating the subpolygons
$P_p$, $P_q$ that contain $p$ and $q$. For this, we shoot upward
vertical rays from $p$ and from $q$, and we find the first edge (or cut edge)
$e_p$ and $e_q$ of $P$ that is hit. This takes linear time.
Then we determine the subpolygons
$P_p$ and $P_q$ that contain $e_p$ and $e_q$. If the edge is a cut edge, this
is immediate. If not, we go through the description of the subpolygons, and
for each edge sequence, we determine in
constant time whether it contains $e_p$ or $e_q$ by comparing indices.
This requires $O(s)$ time, since the total boundary of all  subpolygons
has $O(s)$ pieces.

If $P_p = P_q$, we
apply the constant-work-space method within $P_q$ and are done. Otherwise,
we take the tree that represents the polygon partition, and we find the path
between $P_p$ and $P_q$. Every edge on that path
corresponds to a cut edge that must be crossed by $\pi_{pq}$,
in the same order. For any subpolygon $P_i$ that is traversed by $\pi_{pq}$,
we define the \emph{entrance} of $P_i$ as the cut edge
through which $\pi_{pq}$ enters $P_i$ and the \emph{exit} of $P_i$ as
the cut edge through which $\pi_{pq}$ leaves $P_i$.
The remaining cut edges
are not be crossed by $\pi_{pq}$, so
we treat them as obstacles.

We initialize the triple $(v, r_1, r_2)$ as in AMRW.
Initially, all of $v$, $r_1$ and $r_2$ are contained in the same subpolygon
of $P$, and we call this polygon $P_{\mathrm{curr}}$.
Next, the algorithm enters the \emph{standard situation}.


\paragraph{Standard Situation.}
In the standard situation, both endpoints $r_1$ and $r_2$ are contained
in the same polygon $P_\text{curr}$. The vertex $v$ does not necessarily lie
in $P_\text{curr}$, but since the AMRW-algorithm works with the subpolygon
$P'$ that is cut off by the segments $vr_1$ and $vr_2$, we only need to
deal with the vertices of $P_\text{curr}$ and $v$, see Figure~\ref{fig:standard}.

\begin{figure}[htb]
 \centerline{\includegraphics[scale=1]{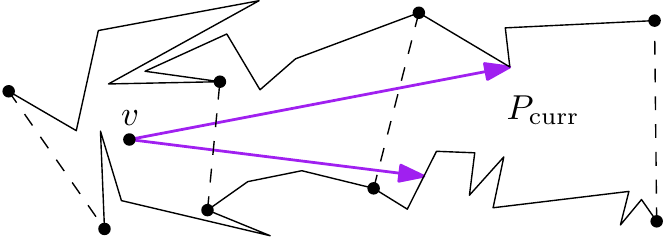}}
\caption{The standard situation.}
 \label{fig:standard}
\end{figure}

We apply the AMRW-strategy almost without change:
we shoot a ray $R'$ that partitions $P_\text{curr}$
into two. If $R'$ does not hit the exit of $P_\text{curr}$, we can use the same
rules as in AMRW to update the triple $(v, r_1, r_2)$ (we take the location of the
exit as a proxy for the target point $q$).

The only problem arises when $R'$ hits the exit.
In this case we must first complete the ray shooting operation:
we extend $R'$ into the adjacent subpolygon. If it again hits the
exit chord of this subpolygon, we continue into the third subpolygon, and so on,
until $R'$ hits a point $p'$ in the boundary of $P$ (or a cut edge that is not
traversed by $\pi_{pq}$). The ray $R'$ splits the wedge defined by $(v, r_1, r_2)$
into two parts, and we take the part that contains the target.
The running time is $O(n/s)$ times the number of subpolygons that are visited.
Then we switch to the \emph{long-jump situation}.
\paragraph{Long-jump situation.}
In general, the invariant of this situation is as follows:
we  maintain a current start vertex $v$ and two
shortest paths, SP$^+$ from $v$ to a point $p^+$
and SP$^-$ from $v$ to a point $p^-$.
The shortest paths form a \emph{funnel} $F$ with apex $v$,
and we maintain the invariant that $q$ lies inside the polygon $P' \subseteq P$
that is cut off by $F$, so $\pi_{pq}$ must go into $F$. Notice that all
funnel vertices are reflex, i.e., they have
interior angle larger than $\pi$.
Let $R_0$ denote some ray from $v$ into $F$.
For the exposition, we assume that $F$ extends from
$v$ to the right, $R_0$ goes in the direction of the positive $x$-axis,
SP$^+$ forms the upper boundary, and SP$^-$ forms the lower boundary, see
Figure~\ref{fig:funnel-1}.

In general, $p^+$ and $p^-$ may lie in different subpolygons $P^+$
and $P^-$. If so, we  assume w.l.o.g.
that $P^+$ is more advanced than $P^-$
(this can be determined in constant time, since we know from the initialization
which sub-polygons must be traversed in which order). Then, our first goal is to
 incrementally extend the shortest path SP$^-$ to the
lower endpoint of the entrance of $P^+$ (Procedure \emph{Catch-up}).

If $P^-=P^+$, we shoot a ray $R'$ and extend one of the
SP edges (Procedure \emph{Extend}). We will proceed in different
ways depending on which side of the ray $R'$ the target $q$
lies (recall that $q$ is the exit of $P^+=P^-$).

\begin{figure}[htb]
\centerline{\includegraphics[scale=0.8]{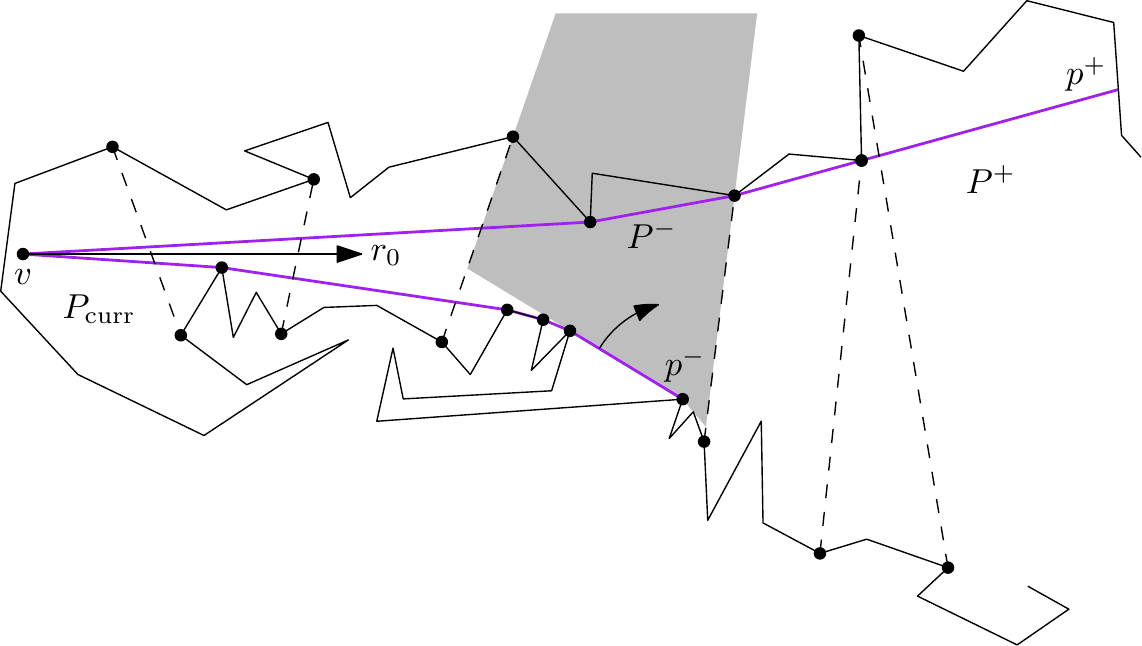}}
\caption{Procedure \emph{Catch-up}: Growing the shorter side of the funnel}
\label{fig:funnel-1}
\end{figure}

\begin{figure}[htb]
\centerline{\includegraphics[scale=0.8]{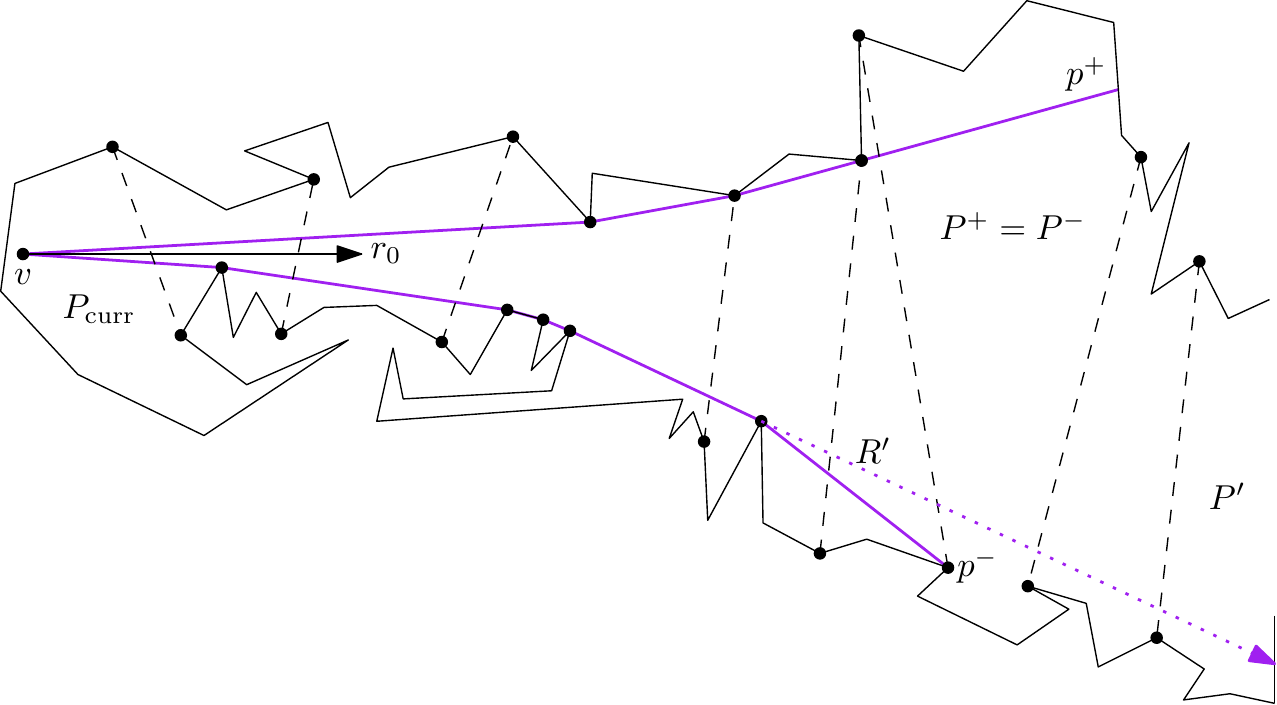}}
\caption{Procedure \emph{Extend}. The ray $R'$ splits the funnel into two funnels.}
\label{funnel-2}
\end{figure}

We have to implement the procedures carefully so that we do not exceed the
storage bounds.  Therefore, we do not store the whole funnel boundaries
SP$^+$ and SP$^-$ explicitly, but we store of each SP only the first edge,
the two last edges, and the $O(s)$ SP edges that cross some cut edge.
We will now give the details of the two procedures.

\noindent
\textbf{Procedure \emph{Catch-up}.}
We have computed the shortest path SP$^-$ to some point
$p^-$ in the subpolygon
$P^-$, and our goal is to proceed towards
the lower endpoint of the exit of $P^-$.
For this, we perform a clockwise sweep of a ray whose initial point is $p^{-}$
and whose original direction follows the last edge of SP$^{-}$, until we
find the first vertex of the lower boundary of $P^{-}$ that is hit,
see Figure~\ref{fig:funnel-1}.

Two cases can happen depending on where the new vertex $w$ lies. If it
makes a right turn with the last edge of SP$^{-}$, we have found
one more edge of the funnel. Thus, we add edge $p^-w$ to the funnel, and
continue. If $w$ makes a left turn, we must remove vertices from $P^{-}$
(so as to satisfy the invariant of the funnel). Those vertices are
removed from $p^-$ towards $v$,
and also, after reaching $v$, from the beginning of SP$^+$ (in the
latter case we output the removed vertices as shortest path vertices).
Finally, we add a new edge to SP$^-$.

Since we do not explicitly store SP$^-$ and SP$^+$, when removing a SP vertex,
we may have to look for the
predecessor (or successor) edges by angular sweeps.
Suppose $e$ is an edge of SP$^-$ whose predecessor we want to determine.
If the predecessor of $e$ crosses the entrance of the current subpolygon,
we can find it in constant time (since it was explicitly stored).
Otherwise, let $z$ be the initial point of $e$.
We perform a counter-clockwise angular sweep of the ray with
start vertex $z$ and whose initial direction is the direction of $e$,
until we hit the first vertex that comes before $z$
on the lower boundary of the relevant subpolygon.  This takes
$O(n/s)$ time for each predecessor or successor that we need to find.
The procedure for finding successor edges is symmetrical.

That is, regardless of where $v$ lies, we can a new edge to the funnel. Once this edge is found, we check where our target lies, and act using the AMRW-strategy. We can repeat this operation until we eventually find the lower end of the exit of $P^{-}$ (for example, after  {\em Catch-up} has been executed in Figure~\ref{fig:funnel-1}). In this case, we advance $P^-$ to the next subpolygon, and we iterate until we reach the entrance of $P^+$.

\noindent
\textbf{Procedure \emph{Extend}.}
We have now two funnel endpoints in the same subpolygon $P^+=P^-$.
If both SP$^+$ and SP$^-$ have only one edge, we can switch back to
the standard situation, letting the triple $(v, r_1, r_2)$ represent the funnel.

Thus, suppose w.l.o.g. that SP$^-$ has more than one edge. We take
the next-to-last edge $e'$ of SP$^-$ and shoot a ray in the direction of
$e'$ into $P^-$, see Figure~\ref{funnel-2}. Call this ray $R'$ and call its origin
vertex $v'$.
If $R'$ does not hit the exit of $P^+$,
the procedure is easy:
we determine, in $O(n/s)$ time, on which side of $R'$ the target lies.
If $q$ lies above $R'$ we pop the last edge of the funnel, as in
Procedure \emph{Catch-up}, and proceed.
If $q$ lies below $R'$ then
we report all vertices along SP$^-$ up to $v'$ (using the same method as
in \emph{Catch-up}).
Then we take the triple that consists of $v'$, the endpoint of $R'$ and
$p^-$, and we switch back to the standard situation (with $P_\text{curr} = P^+$).

Finally, we consider the case that $R'$ goes through the exit,
as shown in Figure~\ref{funnel-2}.
In this case, we extend $R'$ until it hits the boundary of $P$,
in some subpolygon $P'$.
(This takes time
proportional to $n/s$ times the number of subpolygons that are
traversed from $P^+$ to $P'$.)
Next, we determine on which side of $R'$ the target lies.
If $q$ lies above $R'$, $R'$ forms the new last edge of SP$^-$,
$P^-$ is advanced to
$P'$, and we continue with Procedure \emph{Catch-up}.
If the target lies below $R'$, we report all vertices on SP$^+$
up to $v'$, and we take a funnel that consists
only of $R'$ and the last edge of SP$^-$.
We advance $P^+$ to
$P'$ and continue with Procedure \emph{Catch-up}.

\paragraph{Runtime Analysis.}

By the analysis of AMRW, in the standard situation we spend
$O((n/s)^2)$ time 
per subpolygon,
for a total of $O(s(n/s)^2)=O(n^2/s)$.

During the long-jump situation, there are $O(n)$ operations of
adding or removing a vertex of the funnel: each edge is removed at most once,
and even though not all vertices of the funnel are polygon vertices, it still holds
that each edge of $P$ contains at most a constant number of funnel vertices.

Each funnel operation incurs an overhead of  $O(n/s)$ for an
angular sweep or a ray shooting operation.
The only exception occurs when a ray goes through $k+1$ subpolygons,
$k\ge 1$. But in this case the more advanced end of the funnel will
make progress by crossing at least $k$ cut edges. Thus, the running
time $O((k+1) n/s)=O(k\cdot n/s)$ for these cases cannot exceed
$O(s\cdot n/s)=O(n)$ in total.
 We have thus obtained the following theorem.

\begin{theorem}
\label{path-with-space}
  Let $P$ be a simple polygon of $n$ vertices and $s$ be a parameter
  between $1$ and~$n$. We can build a shortest-path-data structure for $P$
  of size $O(s)$ in $O(n^2)$ time and $O(s)$ work-space (or in $O(n)$ time and space).
  With this data structure, we can compute the shortest path
  within $P$ between any two points in $P$ using
  $O(n^2/s)$ time and $O(s)$ work-space.
\qed
\end{theorem}

\section{Open Problems}

Obvious topics for future research are improvements of
the results. For example, it would be interesting obtain a time-space trade-off
for triangulating a \pslg. Furthermore, Theorem~\ref{thm:preprocess} describes
how to find a good cut edge for a simple polygon by essentially triangulating
the whole polygon and giving the most balanced cut. A natural question is if we
can obtain a balanced cut in subquadratic time. The size of the smaller part should
be at least a constant fraction of the whole, but it need not be $1/3$ as in
Proposition~\ref{prop:cut}.  Moreover, the
cut would not necessarily have to be a diagonal connecting two
vertices; any straight (for example, vertical) segment partitioning
the polygon would be fine.

Can the work-space for the last result on finding a shortest path
between $p$ and $q$ (Theorem~\ref{path-with-space}) be reduced to less
than $O(s)$, maybe even constant?  There are two components of our
query algorithm that need $O(s)$ space:
(i)~after locating the subpolygons of $p$ and $q$ in $O(n)$ time, we
 identify and store the sequence of subpolygons traversed by
the path, i.e., we find the path between two vertices in the
tree of subpolygons.
(ii)~we have to process and update the funnels.

There is a chance to reduce the complexity of part~(i): by the
techniques of \cite[Theorem~2]{AsanoMuWa11}, one can walk through the
sequence of subpolygons from $p$ to $q$ in the right order in constant
space and $O(s)$ time.

\bibliographystyle{abbrv}
\bibliography{refs}
\end{document}